\documentclass[runningheads]{llncs}

\usepackage{textcomp}
\usepackage{framed}
\usepackage{amsmath}

\usepackage{amsthm}
\usepackage{mathtools}
\usepackage{caption}
\usepackage[frozencache=true]{minted}
\usepackage[htt]{hyphenat}
\usepackage{xargs}
\usepackage{xparse}
\usepackage{pgfplotstable}
\usepackage{hyperref}
\usepackage{wrapfig}
\usepackage{nicefrac}
\newenvironment{listingsupportingpagebreaks}{\captionsetup{type=listing,singlelinecheck=false,aboveskip=0pt,belowskip=0pt,margin=0pt}}{}

\newcommand\snippet[1]{snippets/#1}

\newcommandx*{\insertlisting}[3]{%
\begin{framed}%
\begin{listingsupportingpagebreaks}%
\inputminted[breaklines,linenos,tabsize=2]{jml.py:JMLLexer -x}{\snippet{#2}}\caption{#1}\label{#3}%
\end{listingsupportingpagebreaks}%
\end{framed}%
}

\newcommandx*{\insertlistingWithoutCaption}[1]{%
\begin{framed}%
\begin{listingsupportingpagebreaks}%
\inputminted[breaklines,tabsize=2]{jml.py:JMLLexer -x}{\snippet{#1}}%
\end{listingsupportingpagebreaks}%
\end{framed}%
}

\newcommandx*{\insertlistingkey}[3]{%
\begin{framed}%
\begin{listingsupportingpagebreaks}%
\inputminted[breaklines,linenos,tabsize=2]{jml.py:KeYLexer -x}{\snippet{#2}}\caption{#1}\label{#3}%
\end{listingsupportingpagebreaks}%
\end{framed}%
}

\newcommandx*{\insertlistingkeyWithoutCaption}[1]{%
\begin{framed}%
\begin{listingsupportingpagebreaks}%
\inputminted[breaklines,tabsize=2]{jml.py:KeYLexer -x}{\snippet{#1}}%
\end{listingsupportingpagebreaks}%
\end{framed}%
}

\newcommandx*{\insertfigure}[4][3=1,4=htpb]{
\begin{figure}[#4]
  \centering
  \includegraphics[width=\columnwidth*\real{#3}]{#2} % Overleaf: defaults to map Figures
  \captionof{figure}{#1}
  \label{#1}
\end{figure}
}

\newcommand\inputjml[1]{\inputminted[breaklines,fontsize=\scriptsize,tabsize=2]{jml.py:JMLLexer -x}{\snippet{#1}}}
\newcommand\inputjmlln[1]{\inputminted[breaklines,fontsize=\scriptsize,linenos=true,tabsize=2]{jml.py:JMLLexer -x}{\snippet{#1}}}
\newcommand\inputkey[1]{\inputminted[breaklines,fontsize=\scriptsize,tabsize=2]{jml.py:KeYLexer -x}{\snippet{#1}}}

\def\bs{\char092}

\pgfplotsset{compat=1.14}
\newtheorem*{prop}{Proposition}

\begin{document}

\title{Verifying OpenJDK's \texttt{LinkedList} using KeY}

\author{{Hans-Dieter} A. Hiep\inst{1} \and Olaf Maathuis\inst{3} \and Jinting Bian\inst{1}\and\\
Frank S. de Boer\inst{1} \and Marko van Eekelen\inst{2} \and Stijn de Gouw\inst{2}}
\authorrunning{H.A. Hiep, O. Maathuis, et al.}

\institute{CWI, Science Park 123, 1098 XG Amsterdam, The Netherlands\\
\email{\{hdh,j.bian,frb\}@cwi.nl}
\and
Open University, P.O. Box 2960, 6401 DL Heerlen, The Netherlands\\
\email{\{marko.vaneekelen,stijn.degouw\}@ou.nl}
\and
Achmea, P.O. Box 700, 7300 HC Apeldoorn, The Netherlands\\
\email{olaf.maathuis@achmea.nl}}

\maketitle

\begin{abstract}
As a particular case study of the formal verification of state-of-the-art, real  software, we discuss the specification and verification of a corrected version of the implementation of a linked list as provided by the Java Collection framework.

\end{abstract}

\section{Introduction}\label{sec:intro}
%Intro

%motivation: real-world verification

%KeY project

%conclusion + related work + intro

Software libraries are the building blocks of millions of programs, and they run on the devices of billions of users every day.
Therefore, their correctness is of the utmost importance.
%Among the arguments that are routinely invoked against the usage of formal software verification one can find the following: it is expensive, it is not worthwhile (compared to its cost), it is less effective than bug finding (e.g., by testing, static analysis, or model checking), it does not work for “real” software.
The importance and potential of formal  software verification as a means of rigorously validating state-of-the-art, real  software and improving it,
%The effectiveness of  formal software verification of real software, i.e., software used in (industrial) practice,
is convincingly illustrated by its application to TimSort, the default sorting library in many widely used programming languages, including Java and Python, and platforms like Android (see \cite{GouwRBBH15,GouwBBHRS19}): a crashing implementation bug was found.

The Java implementation of TimSort belongs to the Java Collection framework which provides
 implementations of basic data structures and is among the most widely used
libraries. Nonetheless, over the years, 877 bugs in the Collections Framework have been reported in the official OpenJDK bug tracker.

%In this paper we describe another application of formal software verification.
Due to the intrinsic complexity of modern software, the possibility of interventions by a human verifier is indispensable
for proving correctness.
This holds in particular for  the Java Collection library,
where programs are expected to behave correctly for inputs of arbitrary size.
As a particular case study, we discuss the formal verification of a corrected version of the implementation of a linked list as specified by the class \verb|LinkedList| of the Java Collection framework. Apart from the fact that the data structure of a linked list is one of the basic structures for storing and maintaining unbounded data, this is an interesting case study because
it provides further evidence that formal verification of real software  can lead to major improvements and correctness guarantees.

\begin{wrapfigure}{L}{0.29\textwidth}
  \centering
  \hspace*{-1em}\includegraphics[scale=0.2]{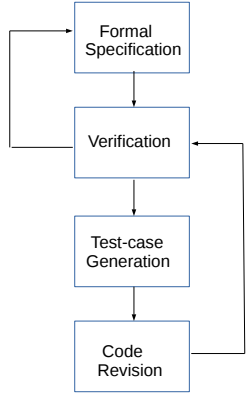}
  \caption{Workflow}
  \label{fig:wf}
\end{wrapfigure}

We follow the general workflow underlying the TimSort case as depicted in Figure \ref{fig:wf}. The workflow starts with a formalisation of the informal documentation of the Java code in the Java Modeling Language \cite{LeavensBR99}. %supported by the KeY theorem prover.
This formalisation goes hand in hand with the formal verification: failed verification attempts can provide information about
further refinements of the specs. A failed verification attempt may also indicate an error in the code,
and can as such be used for the generation of test cases to detect the error at run-time.

\texttt{LinkedList} is the only \texttt{List} implementation in the Collection Framework that allows collections of unbounded size. Following this workflow we found out  that the Java linked list implementation does not correctly take into account
the Java integer overflow semantics.
It is exactly for large lists ($\geq 2^{31}$ items), that the implementation breaks. 
This basic observation gave rise to  a number of test cases which show that Java's \texttt{LinkedList} class breaks 22 methods out of a total of 25 methods of the \texttt{List}!\footnote{We filed a bug report to the official Java bug tracker. Once the report is made public by the Java maintainers, we will add the URL as metadata to our repository~\cite{hiep2019prooffiles}.}

On the basis of these test cases we propose in Section~\ref{sec:linkedlist} also a fixed version  of the Java linked list implementation and formally specify and verify its correctness in Section~\ref{sec:improved}  with respect to the Java integer overflow semantics.
In Section~\ref{sec:discussion} we discuss the main challenges posed by this case study and related work.

This case study has been carried out using the  state-of-the-art KeY theorem prover \cite{KeYbook}, because it formalizes the integer overflow semantics of Java and it allows to directly “load” Java programs.
An archive of proof files and the KeY version used in this work is available \cite{hiep2019prooffiles}.

\section{\texttt{LinkedList} in OpenJDK}\label{sec:linkedlist}

\begin{minipage}{.48\textwidth}
\vspace*{-10pt}
\inputjml{LinkedList_core.java}
\end{minipage}\hspace*{0.04\textwidth}%
\begin{minipage}{.48\textwidth}
\inputjml{add_linkLast.java}
\label{lst:linkLast}
\end{minipage}

\texttt{LinkedList} was introduced in Java version 1.2 as part of Java's Collection Framework in 1998. Figure~\ref{fig:hierarchy} shows how \texttt{LinkedList} fits in the type hierarchy of this framework. \texttt{LinkedList} implements the \texttt{List} interface, and also supports all general \texttt{Collection} methods as well as the methods from the \texttt{Queue} and \texttt{Deque} interface. The \texttt{List} interface provides positional access to the elements of the list, where each element is indexed by Java's primitive \texttt{int} type.

\noindent 
\begin{minipage}{.48\textwidth}
  \centering
  \includegraphics[width=0.95\textwidth]{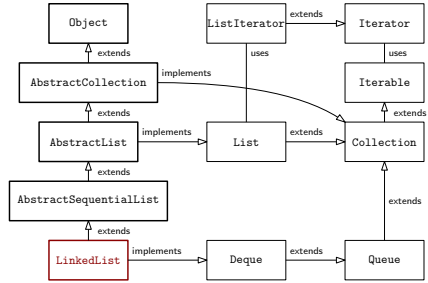}
  \captionof{figure}{Java Collections framework with \texttt{LinkedList} in lower-left corner. Classes have thick edges, interfaces have thin edges. This picture shows the complex inheritance structure.}
  \label{fig:hierarchy}
\end{minipage}\hspace*{0.04\textwidth}%
\begin{minipage}{.48\textwidth}
\vspace*{2pt}The structure of the \texttt{LinkedList} class is shown before.
This class has three attributes: a \texttt{size} field, which stores the number of elements in the list, and two
fields that store a reference to the \texttt{first} and \texttt{last} node. Internally, it uses the private static nested \texttt{Node} class to represent the items in the list. A static nested private class behaves like a top-level class, except that it is not visible outside the enclosing class (\texttt{LinkedList}, in this case). Nodes are doubly linked; each node is connected to the preceding (field \texttt{prev}) and succeeding node (field \texttt{next}). These fields contain \texttt{null} in case no preceding or succeeding node exists. The data itself is contained in the \texttt{item} field of a node.
\end{minipage}

\vspace*{4pt}\texttt{LinkedList} contains 57 methods. Due to space limitations, we focus on three characteristic methods. Method \texttt{add(E)} calls method \texttt{linkLast(E)}, which creates a new \texttt{Node} object to store the new item and adds the new node to the end of the list. Finally the new size is determined by unconditionally incrementing the value of the \texttt{size} field, which has type \texttt{int}. Method \texttt{indexOf(Object)} returns the position (of type \texttt{int}) of the first occurrence of the specified element in the list, or $-1$ if it's not present.

\vspace*{-4pt}\inputjml{indexOf_o_original.java}

%\begin{figure}[t]
  %\centering
  %\includegraphics[scale=0.75]{figures/chain1.eps}
  %\caption{A chain of three nodes, indexed from left to right. Items are not shown.}
  %\label{fig:chain}
%\end{figure}

\vspace*{-4pt}\noindent Each linked list consists of a sequence of nodes. Sequences are finite, indexing of sequences starts at zero, and we write $\sigma[i]$ to mean the $i$th element of some sequence $\sigma$. A \emph{chain} is a sequence $\sigma$ of nodes of length $n>0$ such that: the \texttt{prev} reference of the first node $\sigma[0]$ is \texttt{null}, the \texttt{next} reference of the last node $\sigma[n-1]$ is \texttt{null}, the \texttt{prev} reference of node $\sigma[i]$ is node $\sigma[i-1]$ for every index $0<i<n$, and the \texttt{next} reference of node $\sigma[i]$ is node $\sigma[i+1]$ for every index $0\leq i < n-1$. The \texttt{first} and \texttt{last} references of a linked list are either both \texttt{null} to represent the \emph{empty} linked list, or there is some chain $\sigma$ between the \texttt{first} and \texttt{last} node, viz.~$\sigma[0]=\mathtt{first}$ and $\sigma[n-1]=\mathtt{last}$. Figure \ref{fig:linkedlist} shows example instances.
Also see standard literature such as Knuth's \cite[Section 2.2.5]{knuth1997art}.

\begin{figure}[t]
  \centering
  \includegraphics[scale=0.75]{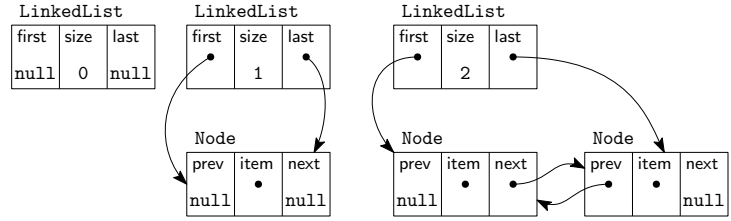}
  \caption{Three example linked lists: empty, with a chain of one node, and with a chain of two nodes. Items themselves are not shown.}
  \label{fig:linkedlist}
\end{figure}

We make a distinction between the \emph{actual} size of a linked list and its \emph{cached} size. In principle, the size of a linked list can be computed by walking through the chain from the \texttt{first} to the \texttt{last} node, following the \texttt{next} reference, and counting the number of nodes. For performance reasons, the Java implementation also maintains a cached size. The cached size is stored in the linked list instance.

Two basic properties of doubly-linked lists are \emph{acyclicity} and \emph{unique first and last nodes}. Acyclicity is the statement that for any indices $0\leq i<j<n$ the nodes $\sigma[i]$ and $\sigma[j]$ are different. First and last nodes are unique: for any index $i$ such that $\sigma[i]$ is a node, the \texttt{next} of $\sigma[i]$ is \texttt{null} if and only if $i=n-1$, and \texttt{prev} of $\sigma[i]$ is \texttt{null} if and only if $i=0$. Each item is stored in a separate node, and the same item may be stored in different nodes when duplicate items are present in the list.

\subsection{Integer overflow bug}\label{sec:bugs}
The size of a linked list is encoded by a signed 32-bit integer (Java's primitive \texttt{int} type) that has a two's complement binary representation where the most significant bit is a sign bit. The values of \texttt{int} are bounded and between $- 2^{31}$ (\texttt{Integer.MIN\_VALUE}) and $2^{31} - 1$ (\texttt{Integer.MAX\_VALUE}), inclusive. Adding one to the maximum value, $2^{31} - 1$, results in the minimum value, $- 2^{31}$: the carry of addition is stored in the sign bit, thereby changing the sign.

Since the linked list implementation maintains one node for each element, its size is implicitly bounded by the number of node instances that can be created. Until 2002, the JVM was limited to a 32-bit address space, imposing a limit of 4 gigabytes (GiB) of memory. In practice this is insufficient to create $2^{31}$ node instances. Since 2002, a 64-bit JVM is available allowing much larger amounts of addressable memory. Depending on the available memory, in principle it is now possible to create $2^{31}$ or more node instances. In practice such lists can be constructed today on systems with 64 gigabytes of memory, e.g., by repeatedly adding elements. However, for such large lists, at least 20 methods break, caused by signed integer overflow. For example, several methods crash with a run-time exception or exhibit unexpected behavior!

Integer overflow bugs are a common attack vector for security vulnerabilities: even if the overflow bug may seem benign, its presence may serve as a small step in a larger attack. Integer overflow bugs can be exploited more easily on large memory machines used for `big data' applications, e.g.~there are real-world attacks that involve Java arrays with approximately $\nicefrac{2^{32}\!}{5}$ elements \cite[Section~3.2]{phrack2018escaping}.

The \texttt{Collection} interface allows for collections with over \texttt{Integer.MAX\_VALUE} elements. For example, its documentation (Javadoc) explicitly states the behavior of the \texttt{size()} method: `Returns the number of elements in this collection. If this collection contains more than \texttt{Integer.MAX\_VALUE} elements, returns \texttt{Integer.MAX\_VALUE}'. The special case (`more than \ldots') for large collections is necessary because \texttt{size()} returns a value of type \texttt{int}.

When \texttt{add(E)} is called and unconditionally increments the \texttt{size} field, an overflow happens after adding $2^{31}$ elements, resulting in a negative \texttt{size} value. In fact, as the Javadoc of the \texttt{List} interface (see Appendix~\ref{sec:java-util-list}) describes, this interface is based on integer indices of elements: `The user can access elements by their integer index (position in the list), \ldots'. For elements beyond \texttt{Integer.MAX\_VALUE}, it is very unclear what integer index should be used. Since there are only $2^{32}$ different integer values,  at most $2^{32}$ node instances can be associated with an unique index. For larger lists, elements cannot be uniquely addressed anymore using an integer index. In essence, as we shall see in more detail below, the bounded nature of the integer indices implies that the design of the \texttt{List} interfaces breaks down for large lists. Remarkably, the actual size of the linked list remains correct as the chain is still in place: most methods of the \texttt{Queue} interface still work. The above observations have many ramifications: it can be shown that 22 out of the 25 methods in the \texttt{List} interface are broken.
\subsection{Reproduction}

%\inputjml{makeListFull.java}

We have run a number of test cases to show the presence of bugs caused by the integer overflow. The running Java version was Oracle's JDK8 (build 1.8.0 201-b09) that has the same \texttt{LinkedList} implementation as in OpenJDK8. Before running a test case, we set up an empty linked list instance. Below, we give an high-level overview of the test cases. Each test case uses \texttt{letSizeOverflow()} or \texttt{addElementsUntilSizeIs0()}: these repeatedly call the method \texttt{add()} to fill the linked list with \texttt{null} elements, and the latter method also adds a last element (\texttt{"this is the last element"}) causing \texttt{size} to be $0$ again.

\begin{enumerate}
\item Directly after \texttt{size} overflows, the \texttt{size()} methods returns a negative value, violating what the corresponding Javadoc stipulates: its value should remain $\texttt{Integer.MAX\_VALUE} = 2^{31} - 1$.

\inputjml{testcase1_output.txt}

Clearly this behavior is in contradiction with the documentation. The actual number of elements is determined by having a field \texttt{count} (of type \texttt{long}) that is incremented each time the method \texttt{add()} is called.

\item The query method \texttt{get(int)} returns the element at the specified position in the list. It throws an \texttt{IndexOutOfBoundsException} exception when \texttt{size} is negative. From the informal specification, it is unclear what indices should be associated with elements beyond \texttt{Integer.MAX\_VALUE}.

\inputjml{testcase2_output.txt}

\item The method \texttt{toArray()} returns an array containing all of the elements in this list in proper sequence (from first to last element). When \texttt{size} is negative, this method throws a \texttt{NegativeArraySizeException} exception. Furthermore, since the array size is bounded by $2^{31}-1$ elements\footnote{In practice, the maximum array length turns out to be $2^{31}-5$, as some bytes are reserved for object headers, but this may vary between Java versions \cite{phrack2018escaping,knuppel2018experience}.}, the contract of \texttt{toArray()} is unsatisfiable for lists larger than this. The method \texttt{Collections.sort(List<T>)} sorts the specified list into ascending order, according to the natural ordering of its elements. This method calls \texttt{toArray()}, and therefore also throws a \texttt{NegativeArraySizeException}.

\inputjml{testcase3_output.txt}

\item Method \texttt{indexOf(Object o)} returns the index of the first occurrence of the specified element in this list, or $-1$ if this list does not contain the element. However due to the overflow, it is possible to have an element in the list associated to index $-1$, which breaks the contract of this method.

\inputjml{testcase4_output.txt}

\item Method \texttt{contains(Object o)} returns true if this list contains the specified element. If an element is associated with index $-1$, it will indicate wrongly that this particular element is not present in the list.

\inputjml{testcase5_output.txt}
\end{enumerate}

Specifically, method \texttt{letSizeOverflow()} adds $2^{31}$ elements that causes the overflow of \texttt{size}. Method \texttt{addElementsUntilSizeIs0()} first adds $2^{32}-1$ elements: the value of \texttt{size} is then $-1$. Then, it adds the last element, and \texttt{size} is $0$ again. All elements added are \texttt{null}, except for the last element. For test cases 4 and 5, we deliberately misuse the overflow bug to associate an element with index $-1$. This means that method \texttt{indexOf(Object)} for this element returns $-1$, which according to the documentation means that the element is not present. For test cases 1, 2 and 3 we needed 65 gigabytes of memory for the JRE on a VM with 67 gigabytes of memory.  For test cases 4 and 5 we needed 167 gigabytes of memory for the JRE on a VM with 172 gigabytes of memory. All test cases were carried out on a machine in a private cloud (SURFsara), which provides instances that satisfy these system requirements.

\subsection{Mitigation}\label{subsec:solutions}
There are multiple directions for mitigating the overflow bug: \emph{do not fix}, \emph{fail fast}, \emph{long size field} and \emph{long or \texttt{BigInteger} indices}. Due to lack of space, we describe only the \emph{fail fast} solution. This solution stays reasonably close to the original implementation of \texttt{LinkedList} and does not leave any behavior unspecified.

%\paragraph{Do not fix} A possible solution is to leave the \texttt{LinkedList} implementation as it is, and specify only part of its behavior. Methods are unspecified in the case in which the real size is larger than \texttt{Integer.MAX\_VALUE}. This implies that every behavior of methods is acceptable when size becomes too large, including the behavior we have observed in Section~\ref{sec:bugs}. The documentation of the implementation could be changed to reflect this choice.
%
%Although this solution requires the least amount of effort, it might lead to surprises to developers who encounter erratic behavior of the implementation. Since part of the behavior is unspecified, this behavior could even change between versions\footnote{Such a change has already occurred in the past: in OpenJDK's 2014 commit, ``Pico-optimize \texttt{contains} method'', the method \texttt{contains(Object)} changed from \texttt{\textbf{return} indexOf(o) != -1;} to \texttt{\textbf{return} indexOf(o) >= 0;} which means that after adding $2^{31}$ elements, the last element (with negative index $-2^{31}$) that used to be contained, no longer is contained after applying this commit.}.

In the \emph{fail fast} solution, we ensure that the overflow of \texttt{size} may never occur. Whenever elements would be added that cause the size field to overflow, the operation throws an exception and leaves the list unchanged. As the exception is triggered right before the overflow would otherwise occur, the value of \texttt{size} is guaranteed to be bounded by \texttt{Integer.MAX\_VALUE}, i.e. it never becomes negative.

This solution requires a slight adaptation of the implementation: for methods that increase the \texttt{size} field, only one additional check has to be performed before a \texttt{LinkedList} instance is modified. This checks whether the result of the method causes an overflow of the \texttt{size} field. Under this condition, an \texttt{IllegalStateException} is thrown. Thus, only in states where \texttt{size} is less than \texttt{Integer.MAX\_VALUE}, it is acceptable to add a single element to the list.

We shall work in a separate class called \texttt{BoundedLinkedList}: this is the improved version that does not allow more than $2^{31}-1$ elements. Compared to the original \texttt{LinkedList}, two methods are added, \texttt{isMaxSize()} and \texttt{checkSize()}:

\inputjml{checkSize.java}

\noindent These methods implement an overflow check. The latter method is called before any modification occurs that increases the size by one: this ensures that \texttt{size} never overflows. Some methods now differ when compared to the original \texttt{LinkedList}, as they involve an invocation of the \texttt{checkSize()} method.

\section{\label{sec:improved}Specification and verification of \texttt{BoundedLinkedList}}

The aim of our specification and verification effort is to verify formalizations of the given Javadoc specifications (stated in natural language) of the \texttt{LinkedList}. This includes establishing absence of overflow errors. Moreover, we restrict our attention only to \texttt{BoundedLinkedList} and not to the rest of the Collection framework or Java classes: methods that involve parameters with interface types, Java serialization or Java reflection are not considered.

%\inputjml{methods_out_of_scope.java}

\texttt{BoundedLinkedList} inherits from \texttt{AbstractSequentialList}, but we consider its inherited methods out of scope. These methods operate on other collections such as \texttt{removeAll} or \texttt{containsAll}, and methods that have other classes as return type such as \texttt{iterator}. However, these methods call methods overloaded by \texttt{BoundedLinkedList}, and can not cause an overflow by themselves.

We have made use of KeY's stub generator to generate dummy contracts for other classes that \texttt{BoundedLinkedList} depends on, such as for the inherited interfaces and abstract super classes. The stub generator moreover deals with generics by erasing the generic type parameters. We assume for exceptions that their constructors are \emph{pure}. An important stub contract is the equality method of the absolute super class \texttt{Object}, which we have adapted: we assume that every object has a \emph{side-effect free}, \emph{terminating} and \emph{deterministic} implementation of its equality method\footnote{In reality, there are Java classes for which equality is not terminating. A nice example is \texttt{LinkedList} itself, where adding a list to itself easily leads to a \texttt{StackOverflowError} when testing equality with a similar instance. We consider this issue out of scope as this behavior is explicitly described by the Javadoc.}.

\inputjml{Object.java}

\subsection{Specification}\label{subsec:specification}
Following our workflow, we have iterated a number of times before the specifications we present here were obtained. This is a costly procedure, as revising some specifications requires redoing most verification effort. Until sufficient information is present in the specification, proving, for example, termination of a method is difficult or even impossible: from stuck verification attempts, and an intuitive idea of why a proof is stuck, the specification is revised.

\emph{Ghost fields.} We use JML's ghost fields: these are logical fields that for each object gets a value assigned in a heap. The value of these fields are conceptual, i.e. only used for specification and verification purposes. During run-time, this field is not present and cannot affect the course of execution. Our improved class is annotated with two ghost fields: \texttt{nodeList} and \texttt{nodeIndex}.

The type of the \texttt{nodeList} ghost field is an abstract data type of sequences, a KeY built-in. This type has standard constructors and operations that can be used in contracts and in JML set annotations. A sequence has a length, which is finite but unbounded. The type of a sequence's length is \texttt{\bs bigint}. In KeY a sequence is unityped: all its elements are of the \emph{any} sort, which can be any Java object reference or primitive, or built-in abstract data type. One needs to apply appropriate casts and track type information for a sequence of elements in order to cast elements of the \emph{any} sort to any of its subsorts.

The \texttt{nodeIndex} ghost field is used as a ghost parameter with unbounded but finite integers as type. This ghost parameter is only used for specifying the behavior of the methods \texttt{unlink(Node)} and \texttt{linkBefore(Object, Node)}. The ghost parameter tracks at which index the \texttt{Node} argument is present in the \texttt{nodeList}. This information is implicit and not needed at run-time. 

\emph{Class invariant.} The ghost field \texttt{nodeList} is used in the class invariant of our improved implementation, see below. We relate the fields \texttt{first} and \texttt{last} that hold a reference to a \texttt{Node} instance, and the chain between \texttt{first} and \texttt{last}, to the contents of the sequence in the ghost field \texttt{nodeList}. This allows us to express properties in terms of \texttt{nodeList}, where they reflect properties about the chain on the heap. One may compare this invariant with the description of chains as given in Section \ref{sec:linkedlist}.

\inputjmlln{invariant.java}

The actual size of a linked list is the length of the ghost field \texttt{nodeList}, whereas the cached size is stored in a 32-bit signed integer field \texttt{size}. On line~4, the invariant expresses that these two must be equal. Since the length of a sequence (and thus \texttt{nodeList}) is never negative, this implies that the size field never overflows. On line~5, this is made explicit: the real size of a linked list is bounded by \texttt{Integer.MAX\_VALUE}. Line~5 is redundant as it follows from line~4, since a 32-bit integer never has a value larger than this maximum value. The condition on lines 6--7 requires that every node in \texttt{nodeList} is an instance of \texttt{Node} which implies it is non-\texttt{null}.

A linked list is either empty or non-empty. On line 8, if the linked list is empty, it is specified that \texttt{first} and \texttt{last} must be \texttt{null} references. On lines 9--12, if the linked list is non-empty, it is specified that \texttt{first} and \texttt{last} are non-\texttt{null} and moreover that the \texttt{prev} field of the first \texttt{Node} and the \texttt{next} field of the last \texttt{Node} are \texttt{null}. The \texttt{nodeList} must have as first element the node pointed to by \texttt{first}, and \texttt{last} as last element.
In any case, but vacuously true if the linked list is empty, does the \texttt{nodeList} form a chain of nodes: lines 13--16 describe that, for every node at index $0 < i < \mathtt{size}$, the \texttt{prev} field must point to its predecessor, and similar for successor nodes.

We note three interesting properties that are implied by the above invariant: acyclicity, unique first and unique last node. These properties can be formulated by JML formulas as follows. These properties are not part of our invariant; instead they introduced interactively in KeY. Otherwise, we need to reestablish these properties each time we show the invariant holds too.

\inputkey{invariant_follows.java}

\begin{figure}
\inputjml{lastIndexOf_on.java}
    \caption{Method \texttt{lastIndexOf(Object)} annotated with JML. Searches the list from last to first for an element. Returns $-1$ if this element is not present in the list; otherwise returns the index of the node that was equal to the argument. Only the contract and branch in which the argument is non-\texttt{null} is shown due to space restrictions. Methods such as \texttt{indexOf}, \texttt{removeFirstOccurrence} and \texttt{removeLastOccurrence} are very similar.}
    \label{lst:indexOf}
\end{figure}

\emph{Methods.} All methods within our scope are given a JML contract that specify its normal behavior and its exceptional behavior. As for an example contract, consider the \texttt{lastIndexOf(Object)} method in Figure~\ref{lst:indexOf}: it searches through the chain of nodes until it finds a node with an item equal to the argument. This method is interesting due to its potential overflow behavior of the resulting index.
\texttt{BoundedLinkedList} together with all method specifications are available \cite{hiep2019prooffiles}.
%Another example is the method \texttt{clear()}. It clears the fields \texttt{item}, \texttt{prev}, and \texttt{next} of all the nodes in the list to aid the garbage collector. Then the value of field \texttt{size} is set to \texttt{0}, and the fields \texttt{first} and \texttt{last} are set to \texttt{null}. The difficulty in this method is to show termination of this loop, as the class invariant of the linked list is broken while the loop is in progress. Its contract and loop invariant are depicted in Figure~\ref{lst:clear}.

\subsection{Verification}\label{subsec:verification}

We start by giving a general strategy we apply to verify proof obligations. We also describe in more detail how to produce a single proof, in this case \texttt{lastIndexOf(Object)}. This gives a general feel how proving in KeY works. This method is neither trivial, nor very complicated to verify. In this manner, we have produced proofs for each method contract that we have specified.

\emph{Overview of verification steps.} When verifying a method, KeY first has to perform symbolic execution. Symbolic execution transforms modal operators on program fragments into \texttt{JavaDL}. We have KeY do this by applying a macro. When doing this, simplification rules are also applied automatically. We keep in mind that the class invariant contains disjunction, and in case we do not want them to be split during the execution of the macro. KeY has to be instructed to delay unfolding the invariant. When symbolic execution is finished, goals may contains heap expressions that must be simplified. When this is done for all heap expressions, the open goals look as simple as they can be before simplifying them further. This in general may be a good moment to compare the open goals to the method and its annotations, and see whether things in KeY look familiar at this point. In the remaining part of the proof the user must find an appropriate mix between interactive and automatic steps.

There are many ways to construct a closed proof tree. At (almost) every step the user has a choice between applying steps manually or automatically. When applying a manual step, different choices can be made which rule in what order to apply where: it takes some experience to choose the best rule. An important rule is the cut rule that splits a proof tree into two parts. The cut rule significantly reduce the size of a proof and the effort required to produce it. For example, the acylicity property can be introduced using cut.

\emph{Verification example.} The method \texttt{lastIndexOf} has two contracts: one involves a \texttt{null} argument, and another involves a non-\texttt{null} argument. Both proofs are similar. Moreover, the proof for \texttt{indexOf(...)} is similar but involves the \texttt{next} reference instead of the \texttt{prev} reference. This contract is interesting, since proving its correctness shows the absence of the overflow bug.

\begin{prop}
\emph{\texttt{lastIndexOf(Object)}} as specified in Figure~\ref{lst:indexOf} is correct.
\end{prop}
\begin{proof}
Set strategy to default strategy, and set max. rules to 5,000, class axiom delayed. Finish symbolic execution on the main goal. Set strategy to 1,000 rules and select DefOps arithmetic rules. Close all provable goals under the root. One goal remains. Perform update simplification macro on the whole sequent, perform propositional with split macro on the consequent with conjunctions, and close provable goals on the resulting subtree. There are two remaining cases:

\begin{itemize}
    \item Case $\mathit{index} - 1 = 0 \leftrightarrow x.\mathtt{prev} = \mathtt{null}$: split the equivalence. First case, suppose $\mathit{index} - 1 = 0$, then $x = \mathit{self}.\mathtt{nodeList}[0]=\mathit{self}.\mathtt{first}$ and\\
    $\mathit{self}.\mathtt{first}.\mathtt{prev}=\mathtt{null}$: solvable through unfolding the invariant and equational rewriting. Now, second case, suppose $x.\mathtt{prev} = \mathtt{null}$. Then, either $\mathit{index} = 1$ or $\mathit{index} > 1$ (from splitting $\mathit{index} \geq 1$). The first of which is trivial (close provable goal), and the second one requires instantiating quantified statements from the invariant, leading to a contradiction. Since we have supposed $x.\mathtt{prev} = \mathtt{null}$, but $x = \mathit{self}.\mathtt{nodeList}[\mathit{index} - 1]$ and $\mathit{self}.\mathtt{nodeList}[\mathit{index} - 1].prev = \mathit{self}.\mathtt{nodeList}[\mathit{index} - 2]$ and\\
    $\mathit{self}.\mathtt{nodeList}[\mathit{index} - 2] \neq \mathtt{null}$.
    \item Case $\mathit{self}.\mathtt{nodeList}[\mathit{index} - 2] = x.\mathtt{prev}$. Follows from invariant, where $\mathit{self}.\mathtt{nodeList}[\mathit{index} - 1] = x$ and\\
    $\mathit{self}.\mathtt{nodeList}[\mathit{index} - 1].\mathtt{prev} = \mathit{self}.\mathtt{nodeList}[\mathit{index} - 2]$, and equational rewriting.
\end{itemize}
\vspace*{-2em}
\end{proof}

\emph{Interesting verification conditions.} The acyclicity property is used to close verification conditions that arise as a result of potential aliasing of node instances: it is used as a separation lemma. Whenever a method changes the \texttt{next} or \texttt{prev} fields of existing node(s) (see e.g. \texttt{linkLast} on page \pageref{lst:linkLast}), we must establish that in the new list, all nodes remain reachable from the first through \texttt{prev} and \texttt{next} (i.e., ``connectedness''). We proved this by using the fact that two nodes instances are different if they have a different index in \texttt{nodeList}, which follows from acyclicity. We sketch an argument why the acyclicity property follows from the invariant. Below we show how the argument in KeY goes, see \cite[0:55--11:30]{Bian2019}.

\begin{prop}
Acyclicity follows from the linked list invariant.
\end{prop}
\begin{proof}
By contradiction: suppose a linked list of size $n > 1$ is not acyclic. Then there are two indices, $0\leq i<j<n$, such that the nodes at index $i$ and $j$ are equal. Then it must hold that for all $j \leq k < n$, the node at $k$ is equal to the node at $k-(j-i)$. This follows from induction. Base case: if $k = j$, then node $j$ and node $j-(j-i)=i$ are equal by assumption. Induction step: suppose node at $k$ is equal to node at $k-(j-i)$, then if $k+1<n$ it also holds that node $k+1$ equals node $k+1-(j-i)$: this follows from the fact that node $k+1$ and $k+1-(j-i)$ are both the \texttt{next} of node $k<n-1$ and node $k-(j-i)$. Since the latter are equal, the former must be equal too. Now, for all $j \leq k < n$, node $k$ equals node $k-(j-i)$ in particular holds when $k=n-1$. However, by the property that only the last node has a \texttt{null} value for \texttt{next}, and a non-last node has a non-\texttt{null} value for its \texttt{next} field, we derive a contradiction: if nodes $k$ and $k-(j-i)$ are equal then all their fields must also have equal values, but node $k$ has a \texttt{null} and node $k-(j-i)$ has a non-\texttt{null} next field!
\end{proof}

%TODO: clear snippets

\emph{Summary of verification effort.} The total effort of our case study was about 7 man months. The largest part of this effort is finding the right specification. KeY supports various ways to specify Java code: model fields/methods, pure methods, and ghost variables. For example, using pure methods, contracts are specified by expressing the content of the list before/after the method using the pure method \text{get(i)}, which returns the $i$\textsuperscript{th} item. This led to rather complex proofs: essentially it led to reasoning in terms of relational properties on programs (i.e. \texttt{get(i)} before vs \text{get(i)} after the method under consideration). After 2.5 man months of writing partial specifications and partial proofs in these different formalisms, we decided to go with ghost variables as this was the only formalism in which we succeeded to prove non-trivial methods.

It then took $\approx$ 4 man months of iterating in our workflow through (failed) partial proof attempts and refining the specs until they were sufficiently complete.
In particular, changes to the class invariant were ``costly'', as this typically caused proofs of all the methods to break (one must prove that all methods preserve the class invariant). The possibility to interact with the prover was crucial to pinpoint the cause of a failed verification attempt, and we used this feature of KeY extensively to find the right changes/additions to the specifications.

%In general there seems to be a hierarchy of specification levels: ghost fields, class invariants, method contracts, and loop invariants. Changing one of a former specification requires revisiting one of the latter. For example, without proper ghost fields it is difficult to express a useful class invariant. Method contracts implicitly depend on the class invariant. For a couple of methods, loop invariants had to be added as annotation. We almost always needed some iterations before finding a loop invariant that was sufficient. Moreover, a change in the method contract often implies any of its loop invariant need to be changed as well, to maintain the properties needed to show the contract is valid.

After the introduction of the field \texttt{nodeList}, several methods could be proved very easily, with a very low number of interactive steps or even automatically. Methods \texttt{unlink(Node)} and \texttt{linkBefore(Object, Node)} could not be proven without knowing the position of the node argument. We introduced a new ghost field, \texttt{nodeIndex}, that acts like a ghost parameter. Luckily, this did not affect the class invariant, and existing proofs that did not make use of the new ghost field were unaffected.

Once the specifications are (sufficiently) complete, we estimate that it only took approximately 1 or 1.5 man weeks to prove all methods. This can be reduced further if informal proof descriptions are given. See Appendix \ref{sec:proofs} for some example descriptions. Moreover, we have recorded a video of a 30 minute proof session where the method \texttt{unlinkLast} is proven correct with respect to its contract \cite{Bian2019}.

\emph{Proof statistics.} The below table summarizes the main proof statistics for all methods. The last two columns are not metrics of the proof, but they indicate the total lines of code (LoC) and the total lines of specifications (LoSpec).

\begin{center}
\begin{tabular}{|c|c|c|c|c|c|c|c|}
\hline
Rules & Branches & Interactive & Quant.ins & Contract & LoopInv & LoC & LoSpec \\
\hline
369,026 & 2,422 & 9,605 & 2,271 & 77 & 12 & 328 & 632 \\
\hline
\end{tabular}
\end{center}

We found the most difficult proofs were for the method contracts of: \texttt{clear()}, \texttt{linkBefore(Object,Node)}, \texttt{unlink(Node)}, \texttt{node(int)} and \texttt{remove(Object)}. The number of interactive steps seem a rough measure for effort required. But, we note that it is not a reliable representation of the difficulty of a proof: an experienced user can produce a proof with very few interactive steps, while an inexperienced user may take many more steps. The proofs we have produced are by no means minimal.

\section{Discussion}\label{sec:discussion}
In this section we discuss some of the main challenges of verifying the real-world Java implementation of a \texttt{LinkedList}, as opposed to the analysis of an idealized mathematical linked list.
%In this section, we discuss several challenges the \texttt{LinkedList} case study poses to analysis tools based on our experiences.
%we discuss several challenges encountered in the \texttt{LinkedList} case study. Many of the challenges are relevant beyond the case study and apply to real-world library/program verification in general.

\emph{Extensive use of Java language constructs.}
The \texttt{LinkedList} class uses a wide range of Java language features. This includes nested classes (both static and non-static), inheritance, polymorphism, generics, exception (handling), object creation and foreach loops.
To load and reason about the real-world \texttt{LinkedList} source code, requires an analysis tool with high coverage of the Java language, including support for the aforementioned language features.
%TODO: new Java release cycle? new feature release every 6 months! how can analysis tools keep up?
% reflection?

\emph{Support for intricate Java semantics.} The Java \texttt{List} interface is position based, and associates with each item in the list an index of Java's int type. The bugs described in Section~\ref{sec:bugs} were triggered on large lists, in which integer overflows occurred. Thus, while an idealized mathematical integer semantics is much simpler for reasoning, it could not be used to analyze the bugs we encountered! It is therefore critical that the analysis tool faithfully supports Java's semantics, including Java's integer (overflow) behavior.

\emph{Collections have a huge state space.} A Java collection is an object that contains other objects (of a reference type). Collections can typically grow to an arbitrary (but typically bounded) size. By their very nature, collections thus intrinsically have a large state. To make this more concrete: triggering the bugs in \texttt{LinkedList} requires at least $2^{31}$ elements (and 64 GiB of memory), and each element, since it is of a reference type, has at least $2^{32}$ values. This poses serious problems to fully automated analysis methods that explore the state space.
\emph{Interface specifications.} Several of the \texttt{LinkedList} methods contain an interface type as parameter, for example the \texttt{addAll} method:
\inputjml{addAll_c_original.java}
As KeY follows the design by contract paradigm, verification of \texttt{LinkedList}'s \texttt{addAll} method requires a contract for each of the other methods called, including the \texttt{toArray} method in the \texttt{Collection} \emph{interface}. Hence, the question arises: how can we specify interface methods, such as \texttt{Collection.toArray}?
Simple conditions on parameters or the return value are easily expressed, but meaningful contracts about the contents of the collection require some notion of state to capture all mutations of the collection, so that previous calls to methods in the interface that contributed to the current content of the collection are taken into account.
Model fields/methods are a widely used mechanism to define an abstract state, given by one or more model variables, in terms of the concrete state given (by the fields) in a concrete class.
In this case, as only the interface type \texttt{Collection} is known rather than a concrete class, a represents clause cannot be defined. Thus the behavior of the interface cannot be fully captured by specifications in terms of model fields/variables, including for methods such as \texttt{Collection.toArray}. Ghost variables cannot be used either, since ghost variables are updated by adding set statements in method bodies, and interfaces do not contain method bodies. This raises the question: how to specify behavior of interface methods?\footnote{Since the representation of classes that implement the interface is unknown in the interface itself, a particularly challenging aspect here is: how to specify the footprint of an interface method, i.e.: what part of the heap can be modified by the method in the implementing class?}
%TODO (RE)MOVE solution?

\emph{Verifiable code revisions.} We fixed the \texttt{LinkedList} class by explicitly bounding its maximum size to \texttt{Integer.MAX\_Value} elements, but other solutions are possible.
Rather than using integers indices for elements, one could change to an index of type \texttt{long} or \text{BigInteger}. Such a code revision is however incompatible with the general \texttt{Collection} and \texttt{List} interfaces (whose method signatures mandate the use integer indices), thereby breaking all existing client code that uses \texttt{LinkedList}.
Clearly this is not an option in a widely used language like Java, or any language that aims to be backwards compatible.

It raises the challenge: can we find code revisions that are compatible with existing interfaces and client classes? We can take this challenge even further:
can we use our workflow to find such compatible code revisions, \emph{and are also amenable to formal verification?} The existing code in general is not designed for verification. For example, the \texttt{LinkedList} class exposes several implementation details to classes in the \text{java.util} package: i.e., all fields, including \texttt{size}, are package private (not private!), which means they can be assigned a new value directly (without calling any methods) by other classes in that package. This includes setting \texttt{size} to negative values (!). As we have seen, the class malfunctions for negative \texttt{size} values.
In short, this means that the \texttt{LinkedList} itself cannot enforce its own invariants anymore: its correctness now depends on the correctness of other classes in the package. The possibility to avoid calling methods to access the lists field may yield a small performance gain, but it precludes a modular analysis: to assess the correctness of \texttt{LinkedList} one must now analyze all classes in the same package (!) to determine whether they make benign changes (if any) to the fields of the list. Hence, we recommend to encapsulate such implementation details, including making at least all fields \texttt{private}.

\emph{Proof reuse.} Section~\ref{subsec:verification} discussed the proof effort (in person months). It revealed that while the total effort was 6-7 person months, once the specifications are in place after many iterations of the workflow, producing the actual final proofs took only 1-2 weeks! But minor specification changes often require to redo nearly the whole proof, causing an explosion in the amount of effort needed.
Other program verification case studies~\cite{Baumann12,GouwBR14,GouwRBBH15,KeYbook} show similarly that the main bottleneck today is specification, not verification.
This calls for techniques to optimize proof reuse.

\emph{Status of the challenges.} Our case study raised various challenges for mechanized verification.
The KeY system covered the Java language features sufficiently to load and statically verify the \texttt{LinkedList} source code. KeY also supports various integer semantics, allowing us to analyze \text{LinkedList} with the actual Java integer overflow semantics.
As KeY is a theorem prover (based on deductive verification), it does not explore the state space of the class under consideration, thus solving the problem of the huge state space of Java collections. We could not find any other tools that solved these challenges (see also related work below), so we decided at that point to use KeY.

Most of the other challenges are still open. The challenge concerning ``Interface specifications'' could perhaps be addressed by defining an abstract state of an interface by using/developing some form of a trace specification that map a sequence of calls to the interface methods to a value, together with a logic to reason about such trace specifications.

The challenges related to code revisions and proof reuse are compounded for analysis tools that use very fine-grained proof representations.
For example, proofs in KeY consist of actual rule applications (rather than higher level macro/strategy applications), and proof rule applications explicitly refer to the indices of the (sub) formulas the rule is applied to.
This results in a fragile proof format, where small changes to the specifications or source code (such as a code refactoring) break the proof.

Other state-of-the-art systems such as Coq, Isabelle and PVS support proof \emph{scripts}. Those proofs are described at a typically much more course-grained level when compared to KeY. It would be interesting to see to what extent Java language features and semantics can be handled in (extensions of) such higher level proof script languages.

\subsection{\label{sec:related}Related work}
Kn{\"u}ppel et al.~\cite{knuppel2018experience} provide a report on the specification and verification
of some methods of the classes ArrayList, Arrays, and Math of the OpenJDK Collections framework using KeY.
Their report is mainly meant as a  ``stepping stone towards a case study for future research.''
To the best of our knowledge,  no formal specification and verification of the actual Java implementation of a linked list has been investigated. 
In general, the data structure of a linked list has been studied mainly in terms of pseudo code of an idealized mathematical abstraction (see \cite{polikarpova2015fully} for an Eiffel version and \cite{KlebanovMSLWAABCCHJLMPPRSTTUW11} for a Dafny version).

This paper (and \cite{knuppel2018experience}) has shown that the specification and verification of actual library software poses a number of serious challenges to formal verification. In our case study, we used KeY to verify Java's linked list. Other systems are also used to formalize Java, such as the general purpose theorem prover Isabelle/HOL~\cite{Nipkow:1998:JLT:268946.268960,Klein:2006:MMJ:1146809.1146811}, or OpenJML~\cite{Cok14}, a prover dedicated to Java programs. However, these formalizations do not have a complete enough Java semantics to be able to analyze the bugs presented in this paper. In particular, these formalizations have no support for integer overflow arithmetic.

% TODO? OpenJDK bug tracker (2003)

%Other linkedlist verifications: Dafny linkedlist, A fully verified container library, Polikarpova et al \cite{polikarpova2015fully}

%Isabelle/HOL Java-level formalizations: JavaLight, otherwise known as Bali, do not use modulo semantics for integers, but instead use mathematical integers \cite{Nipkow:1998:JLT:268946.268960}. The other well-known Isabelle/HOL formalization, Ninja, is not realistically modelling Java: it also does not use modulo semantics for integers \cite{Klein:2006:MMJ:1146809.1146811}.

%Other JVM-level formalizations: Executable JVM model for analytical reasoning: A study (Liu, Moore), CoJaq: a hierarchical view on the Java bytecodeformalised in Coq (Czarnik et al)

%Other specification languages related to JML: CASL specification language (Sannella, Tarlecki), K Framework (Rosu), Alloy (Jackson)

\newpage
\bibliographystyle{splncs04}
\bibliography{main}

\begin{thebibliography}{10}
\providecommand{\url}[1]{\texttt{#1}}
\providecommand{\urlprefix}{URL }
\providecommand{\doi}[1]{https://doi.org/#1}

\bibitem{KeYbook}
Ahrendt, W., Beckert, B., Bubel, R., H{\"{a}}hnle, R., Schmitt, P.H., Ulbrich,
  M. (eds.): Deductive Software Verification - The KeY Book - From Theory to
  Practice, Lecture Notes in Computer Science, vol. 10001. Springer (2016)

\bibitem{Baumann12}
Baumann, C., Beckert, B., Blasum, H., Bormer, T.: Lessons learned from
  microkernel verification -- specification is the new bottleneck. In:
  Proceedings Seventh Conference on Systems Software Verification, {SSV} 2012,
  Sydney, Australia, 28-30 November 2012. pp. 18--32 (2012).
  \doi{10.4204/EPTCS.102.4}

\bibitem{Bian2019}
Bian, J., Hiep, H.A.: {Verifying OpenJDK's LinkedList using KeY: Video} (2019).
  \doi{10.6084/m9.figshare.10033094.v2}

\bibitem{Cok14}
Cok, D.R.: {OpenJML}: Software verification for {Java 7} using {JML},
  {OpenJDK}, and {Eclipse}. In: Proceedings 1st Workshop on Formal Integrated
  Development Environment, {F-IDE} 2014, Grenoble, France, April 6, 2014. pp.
  79--92 (2014). \doi{10.4204/EPTCS.149.8}

\bibitem{phrack2018escaping}
Eauvidoum, I., disk noise: Twenty years of escaping the {Java} sandbox. Phrack
  Magazine  (September 2018),
  \url{http://www.phrack.org/papers/escaping_the_java_sandbox.html}

\bibitem{GouwBBHRS19}
de~Gouw, S., de~Boer, F.S., Bubel, R., H{\"{a}}hnle, R., Rot, J.,
  Steinh{\"{o}}fel, D.: Verifying {OpenJDK}'s sort method for generic
  collections. J. Autom. Reasoning  \textbf{62}(1),  93--126 (2019)

\bibitem{GouwBR14}
de~Gouw, S., de~Boer, F.S., Rot, J.: {Proof} {Pearl}: The {KeY} to correct and
  stable sorting. J. Autom. Reasoning  \textbf{53}(2),  129--139 (2014)

\bibitem{GouwRBBH15}
de~Gouw, S., Rot, J., de~Boer, F.S., Bubel, R., H{\"{a}}hnle, R.: {OpenJDK's
  Java.utils.Collection.sort()} is broken: The good, the bad and the worst
  case. In: Computer Aided Verification - 27th International Conference, {CAV}
  2015, San Francisco, CA, USA, July 18-24, 2015, Proceedings, Part {I}. pp.
  273--289 (2015)

\bibitem{hiep2019prooffiles}
Hiep, H.A., Maathuis, O., Bian, J., de~Boer, F.S., van Eekelen, M., de~Gouw,
  S.: {Verifying OpenJDK's LinkedList using KeY: Proof Files} (2019).
  \doi{10.5281/zenodo.3517082}

\bibitem{KlebanovMSLWAABCCHJLMPPRSTTUW11}
Klebanov, V., M{\"{u}}ller, P., et~al.: The 1st verified software competition:
  Experience report. In: {FM} 2011: Formal Methods - 17th International
  Symposium on Formal Methods, Limerick, Ireland, June 20-24, 2011.
  Proceedings. pp. 154--168 (2011). \doi{10.1007/978-3-642-21437-0\_14}

\bibitem{Klein:2006:MMJ:1146809.1146811}
Klein, G., Nipkow, T.: A machine-checked model for a {J}ava-like language,
  virtual machine, and compiler. ACM Trans. Program. Lang. Syst.
  \textbf{28}(4),  619--695 (Jul 2006). \doi{10.1145/1146809.1146811}

\bibitem{knuppel2018experience}
Kn{\"u}ppel, A., Th{\"u}m, T., Pardylla, C., Schaefer, I.: Experience report on
  formally verifying parts of {OpenJDK's API} with {KeY}. arXiv preprint
  arXiv:1811.10818  (2018)

\bibitem{knuth1997art}
Knuth, D.E.: The art of computer programming, vol.~1. Addison-Wesley (1997)

\bibitem{LeavensBR99}
Leavens, G.T., Baker, A.L., Ruby, C.: {JML:} {A} notation for detailed design.
  In: Behavioral Specifications of Businesses and Systems, pp. 175--188 (1999)

\bibitem{Nipkow:1998:JLT:268946.268960}
Nipkow, T., von Oheimb, D.: Java\emph{light} is type-safe: definitely. In:
  Proceedings of the 25th ACM SIGPLAN-SIGACT Symposium on Principles of
  Programming Languages. pp. 161--170. POPL '98, ACM (1998).
  \doi{10.1145/268946.268960}

\bibitem{polikarpova2015fully}
Polikarpova, N., Tschannen, J., Furia, C.A.: A fully verified container
  library. In: International Symposium on Formal Methods. pp. 414--434.
  Springer, Cham (2015)

\end{thebibliography}

\appendix

\newpage\section{\label{sec:proofs}Proof Descriptions}

\begin{table*}
    \centering
    \begin{filecontents*}{statistics/proof_statistics.csv}
Visibility;Method;Nodes;Branches;Interactive steps;Quantifier instantiations;Operation Contract apps;Loop invariant;LOC;LOS
private;checkElementIndex(int);1,373;14;0;8;1;0;4;5
private;checkPositionIndex(int);1,495;14;0;10;1;0;6;5
private;checkSize();1,455;14;0;10;1;0;4;3
private;linkFirst(Object);33,651;180;665;277;0;0;9;8
private;unlinkFirst(Node);11,187;114;276;30;0;0;13;7
private;unlinkLast(Node);11,024;101;252;31;0;0;13;8
package;linkBefore(Object,Node);28,147;222;531;109;0;0;10;13
package;linkLast(Object);24,576;131;33;256;0;0;9;8
package;node(int);12,953;54;226;21;0;2;13;23
package;unlink(Node);65,800;261;1,604;426;0;0;19;11
public;add(int,Object);5,577;41;423;14;5;0;6;10
public;add(Object);1,676;11;185;7;2;0;5;8
public;addFirst(Object);1,864;10;31;6;2;0;4;7
public;addLast(Object);1,694;10;159;6;2;0;4;7
public;BoundedLinkedList();563;3;0;0;0;0;1;2
public;clear();16,341;130;1,068;35;0;1;12;42
public;contains(null);961;12;51;3;1;0;4;12
public;contains(Object);717;13;10;9;2;0;4;12
public;element();2,045;27;0;28;1;0;3;5
public;get(int);3,108;21;37;9;2;0;4;5
public;getFirst();566;7;0;4;0;0;5;5
public;getLast();575;7;0;4;0;0;5;5
public;indexOf(null);7,115;49;99;10;0;1;12;30
public;indexOf(Object);6,582;43;56;9;1;1;12;30
public;lastIndexOf(null);4,145;29;51;6;0;1;12;30
public;lastIndexOf(Object);4,572;26;23;7;1;1;12;30
public;offer(Object);2,949;16;2;11;1;0;3;8
public;offerFirst(Object);2,813;18;2;12;1;0;4;8
public;offerLast(Object);2,011;10;2;6;1;0;4;8
public;peek();943;8;0;5;0;0;4;4
public;peekFirst();943;8;0;5;0;0;4;4
public;peekLast();944;8;0;5;0;0;4;4
public;poll();7,257;28;0;46;1;0;4;4
public;pollFirst();21,717;115;0;362;1;0;4;5
public;pollLast();3,069;13;1;2;1;0;4;5
public;pop();10,196;40;4;105;1;0;3;6
public;push(Object);4,216;24;3;17;2;0;3;7
public;remove(int);1,863;26;958;4;3;0;4;8
public;remove(null);7,043;72;201;15;1;1;11;36
public;remove(Object);4,217;63;211;13;2;1;11;36
public;removeFirst();633;8;11;1;1;0;5;6
public;removeFirstOccurrence();968;12;60;3;1;0;3;15
public;removeFirstOccurrence(Object);699;14;6;3;1;0;3;15
public;removeLast();10,841;25;10;63;1;0;5;7
public;removeLastOccurrence();4,226;36;27;9;1;1;11;34
public;removeLastOccurrence(Object);4,498;38;26;9;2;1;11;34
public;set(int,Object);1,510;29;41;17;4;0;8;7
public;size();86;4;32;0;0;0;3;3
public;toArray();21,222;102;420;115;1;1;7;37
;Total;364,626;2,261;7,797;2,163;49;12;328;632
\end{filecontents*}
\pgfplotstabletypeset[
    col sep=semicolon,
    string type,
    columns/Visibility/.style={column name=Visibility, column type={|l}},
    columns/Method/.style={column name=Method, column type={|l}},
    columns/Nodes/.style={column name=Nodes, column type={|r}},
    columns/Branches/.style={column name=Br., column type={|r}},
    columns/Interactive steps/.style={column name=I.steps, column type={|r}},
    columns/Quantifier instantiations/.style={column name=Q.ins, column type={|r}},
    columns/Operation Contract apps/.style={column name=C., column type={|r}},
    columns/Loop invariant/.style={column name=li, column type={|r}},
    columns/LOC/.style={column name=loc, column type={|r}},
    columns/LOS/.style={column name=los, column type={|r|}},
    every head row/.style={before row=\hline,after row=\hline},
    every last row/.style={before row=\hline,after row=\hline},
    ]{statistics/proof_statistics.csv}
    \vspace{4pt}
    \caption{Proof statistics of normal behavior contracts. See Table \ref{tab:proof_statistics_exc} for meaning of column abbreviations.}
    \label{tab:proof_statistics_normal}
\end{table*}

\begin{table*}
    \centering
    \begin{filecontents*}{statistics/proof_statistics_2.csv}
Visibility;Method;Exception;Nodes;Branches;Interactive steps;Quantifier instantiations;Operation Contract apps;Loop invariant;LOC;LOS
private;checkElementIndex(int);IOOBE;295;9;184;2;3;0;6;5
private;checkPositionIndex(int);IOOBE;280;10;160;2;3;0;6;5
private;checkSize();ISE;156;7;90;2;2;0;4;4
public;add(int,Object);ISE;666;11;79;9;1;0;6;5
public;add(int,Object);IOOBE;174;8;89;5;2;0;6;7
public;add(Object);ISE;135;6;54;5;1;0;5;5
public;addFirst(Object);ISE;134;6;53;5;1;0;4;5
public;addLast(Object);ISE;136;6;55;5;1;0;4;5
public;element();NSEE;147;6;74;5;1;0;3;5
public;get(int);IOOBE;159;6;82;6;1;0;4;5
public;getFirst();NSEE;220;10;77;5;1;0;5;5
public;getLast();NSEE;218;10;58;5;1;0;5;5
public;offer(Object);ISE;174;7;91;6;1;0;3;5
public;offerFirst(Object);ISE;172;6;90;6;1;0;4;5
public;offerLast(Object);ISE;172;6;90;6;1;0;4;5
public;pop();NSEE;156;6;75;6;1;0;3;5
public;push(Object);ISE;172;6;90;6;1;0;3;5
public;remove();NSEE;156;6;75;6;1;0;3;5
public;remove(int);NSEE;148;6;65;5;1;0;4;6
public;removeFirst();NSEE;227;10;64;5;1;0;5;5
public;removeLast();NSEE;155;7;49;1;1;0;5;5
public;set(int,Object);NSEE;148;6;64;5;1;0;8;5
;;Total;4,400;161;1,808;108;28;0;100;112
\end{filecontents*}
\pgfplotstabletypeset[
    col sep=semicolon,
    string type,
    columns/Visibility/.style={column name=Visibility, column type={|l}},
    columns/Method/.style={column name=Method, column type={|l}},
    columns/Exception/.style={column name=Exc, column type={|l}},
    columns/Nodes/.style={column name=Nodes, column type={|r}},
    columns/Branches/.style={column name=Br., column type={|r}},
    columns/Interactive steps/.style={column name=I.steps, column type={|r}},
    columns/Quantifier instantiations/.style={column name=Q.ins, column type={|r}},
    columns/Operation Contract apps/.style={column name=C., column type={|r}},
    columns/Loop invariant/.style={column name=li, column type={|r}},
    columns/LOC/.style={column name=loc, column type={|r}},
    columns/LOS/.style={column name=los, column type={|r|}},
    every head row/.style={before row=\hline,after row=\hline},
    every last row/.style={before row=\hline,after row=\hline},
    ]{statistics/proof_statistics_2.csv}
    \vspace{4pt}
    \caption{Proof statistics of exceptional behavior contracts. Column abbreviations: Exc is the exception thrown, Br. is the number of proof branches, I.steps is the number of interactive steps, Q.ins is the number of quantifier instantiations, C. is the number of method contracts applied, li is the number of loop invariants, loc are lines of code, and los are lines of specification. Exception abbreviations: IOOBE is index out of bounds exception, ISE is illegal state exception, NSEE is no such element exception}
    \label{tab:proof_statistics_exc}
\end{table*}

This section describes informally how some proofs can be produced using KeY version 2.6.3. This version of KeY and the proof files are available on-line \cite{hiep2019prooffiles}.

To produce proofs in KeY, the first step is to set-up KeY's taclet base to make use of particular groups of rules that correctly model Java's integer overflow semantics. This has to be done only once per set-up, as these taclet settings are stored per computer user. To do so, start KeY and load an example (File, Load Example, choose the first example). Then open Options, Taclet Options, and configure them as follows:

\begin{description}
\item [JavaCard] Off
\item [Strings] On
\item [Assertions] Safe
\item [BigInt] On
\item [Initialization] Disable static initialization
\item [Integer Rules] Java semantics
\item [Integer Simplification Rules] Full
\item [Join Generate Is Weakening Goal] Off
\item [Model Fields] Treat as axiom
\item [More Sequence Rules] On
\item [Permissions] Off
\item [Program Rules] Java
\item [Reach] On
\item [Runtime Exceptions] Ban
\item [Sequences] On
\item [Well-definedness Checks] Off
\item [Well-definedness Operator] L
\end{description}

After setting these taclet options, they become effective after loading the next problem. We do that now: the main proof file \texttt{LinkedList.key} can be loaded, and a contract selection window opens up. Each of the examples below correspond to one contract in this window.

In the following proof descriptions, not all steps are explicitly described, but they are not too difficult to find out. Moreover, variable names may depend on the past session and its user interaction, so they may be different in KeY than what is written here. What is important is that KeY's automated strategy has difficulty solving goals that involves variable casts. A typical approach is to manually instantiate part of the invariant that states that \texttt{nodeList}'s elements are an \texttt{instanceof} the \texttt{Node} class, and applying the rule for narrowing types to bring equations in a similar shape to close the goal. These steps are typically required at the end of a proof.

\newpage\subsection{\texttt{contains(Object)}}

The essence of this proof is to choose the right method contract when \texttt{indexOf} is called. The strategy and symbolic execution go as before, but this time we need to prune the proof right before where the method \texttt{indexOf} is called. We manually choose the right contract: when proving the contract where the argument is \texttt{null}, the corresponding contract for \texttt{indexOf} must be selected: similar for the non-\texttt{null} contract. Then, after continuing symbolic execution and closing the provable goals, there are two cases that remain:

\begin{itemize}
    \item If the result of \texttt{indexOf} is $-1$, then we need to show that all items are not-\texttt{null} (or not equal). This follows from the post-condition of the method contract.
    \item If the result of \texttt{indexOf} is not $-1$, then we need to instantiate the existential quantifier at some point with the result that \texttt{indexOf} has returned.
\end{itemize}

\subsection{\texttt{removeLastOccurrence(null)}}

This method also has two contracts, just like \texttt{lastIndexOf}. Again, the first steps of the proof are the same: setting default strategy and symbolic execution. The remaining cases are different:

\begin{itemize}
    \item Case where \texttt{if} in loop body evaluates to \texttt{false}. Perform update simplification, propositional with split, and close provable goals macros. One goal remains: $\mathit{index} - 1 = 0 \leftrightarrow x.\mathtt{prev} = \mathtt{null}$. Proof is similar as before.
    \item Case where \texttt{if} in loop body evaluates to \texttt{true}. Perform update simplification, propositional with split, and close provable goals macros. One goal remains: introduce $\mathit{index} - 1$ as witness to the existential quantifier. Perform propositional with split, close provable goals macros. Introduce universal quantifier, and the rest follows from loop invariant.
\end{itemize}

\subsection{\texttt{clear()}}

We take the method \texttt{clear()} as an example where the acyclicity property is needed for verification. Figure~\ref{lst:clear} shows that for each node, its fields \texttt{prev}, \texttt{item}, and \texttt{next} are cleared. The maintaining clauses are loop invariants, where ghost field \texttt{index} is used to point to the position in the loop. With every iteration the index is incremented with 1. The second last invariant specifies that for an iteration, the node whose fields are cleared, is the one next to the one that was on turn during the previous iteration. In the proof tree that the user builds up in KeY, at some point this results in a corresponding proof obligation that is partially shown here:

\inputkey{clear_null_1.txt}

\noindent Above is a proof obligation for \texttt{clear()} for the loop invariant where fields are cleared. $index\_1\_0$ is a logical variable derived by KeY, whose meaning is reflected by program variable \texttt{index}. Logical variable $i\_6$ comes is in when the loop invariant is `unpacked' (by applying rule \texttt{allRight}).

\begin{figure}
    \inputjml{clear.java}
    \caption{Method \texttt{clear()} annotated with JML. Removes all of the elements from this list. The list will be empty after this call returns.}
    \label{lst:clear}
\end{figure}

When $i\_6 < index\_1\_0$, how can we be sure that this means that the nodes are different? To be sure, we have to ensure that the list is acyclic. This is expressed by the formula shown below.

\inputkey{acyclic_clear.txt}

\noindent We associate a cut rule with this formula. The goal where the formula is a succedent needs to be closed in order to proof acyclicity. Applying twice \texttt{allRight} and then \texttt{impRight} leaves us with the proof obligation partially shown below.

\inputkey{acyclic_clear_unrolled.txt}

\noindent By taking the opposite of above acyclicity property, viz., by replacing each \texttt{\textbackslash forall} by \texttt{\textbackslash  exists}, and by removing the negation symbol \texttt{\text{!}}, we create a contradiction. We take the formula that is shown below for a next cut rule to be instantiated.

\inputkey{contradict2_acyclic.txt}

\noindent The goal where this formula is on the right side of the sequent, can be closed by using induction. The base case ($k=j\_0$) is evident, and the step case ($k+1$) is based on the 4\textsuperscript{th} maintaining clause of the loop invariant (Figure~\ref{lst:clear}). The goal where this formula is on the left side of the sequent can also be closed. This is because there is a contradiction between \texttt{last} having a \texttt{next} value of \texttt{null}, and at the same time the list being cyclic, i.e., \texttt{last} is the same node as another node not having a \texttt{next} value of \texttt{null}. What makes it closeable in the end is having the formula below on both sides of the sequent. This formula can be deduced by having \texttt{k = self.nodeList.length - 1}.

\inputkey{closing.txt}

\noindent This means acyclicity has been proved and thus can be used as an assumption, in the proof obligation where we needed it in the first place.

\section{Additional listings}\label{sec:add-listings}

\subsection{java.util.List}\label{sec:java-util-list}
Based on OpenJDK8, version jdk8-b132. This version is of 2014-03-04 as found on \url{https://hg.openjdk.java.net/jdk8/jdk8/jdk/}.

Only the essential parts of Javadoc are shown.

\inputminted[breaklines,fontsize=\scriptsize,tabsize=2]{jml.py:JMLLexer -x}{\snippet{List.java}}

\end{document}